\documentclass[conference]{IEEEtran}
\IEEEoverridecommandlockouts

\usepackage{cite}
\usepackage{amsmath,amssymb,amsfonts}
\usepackage{graphicx}
\usepackage{textcomp}
\usepackage{xcolor}
\def\BibTeX{{\rm B\kern-.05em{\sc i\kern-.025em b}\kern-.08em
    T\kern-.1667em\lower.7ex\hbox{E}\kern-.125emX}}

\usepackage[linesnumbered,ruled,vlined]{algorithm2e}
\usepackage{mathrsfs}
\usepackage{amsmath,amsfonts,amsthm}
\usepackage{bm}
\usepackage{colortbl}
\usepackage{cases}
\usepackage{bbding}
\usepackage{pifont}
\usepackage{subfigure}
\usepackage{xcolor}
\usepackage{tabularx}
\usepackage{booktabs}
\usepackage{threeparttable}
\usepackage{graphicx}
\usepackage[export]{adjustbox}
\usepackage{hyperref}
\usepackage{epstopdf}
\usepackage{ulem}
\usepackage{hhline}
\usepackage[OT1]{fontenc}
\usepackage[noend]{algpseudocode}
\usepackage{longtable}
\usepackage{soul}
\usepackage{wasysym}
\usepackage{float} 
\usepackage{tikz}
\newtheorem{theorem}{Theorem}
\newtheorem{lemma}{Lemma}

\newcommand{\circlednumber}[1]{%
  \tikz[baseline=(char.base)]{
    \node[shape=circle, draw=black, fill=white, text=black, inner sep=0.3pt, font=\small] (char) {\sffamily\bfseries #1};%
  }%
}

\definecolor{BrickRed}{RGB}{178,34,34}

\makeatletter  
\newif\if@restonecol  
\makeatother

\usepackage[
	n,
	operators,
	advantage,
	sets,
	adversary,
	landau,
	probability,
	notions,	
	ff,
	mm,
	primitives,
	events,
	complexity,
	asymptotics,
	keys]{cryptocode}

\begin{document}


\title{PIR-DSN: A Decentralized Storage Network Supporting Private Information Retrieval
\thanks{\textsuperscript{\Envelope} Corresponding author: Minghui Xu (\href{mailto:mhxu@sdu.edu.cn}{mhxu@sdu.edu.cn}).}
}

\author{
	\IEEEauthorblockN{Jiahao Zhang$^{\dag}$, Minghui~Xu$^{\dag}$\textsuperscript{\Envelope}, Hechuan Guo$^{\dag}$, Xiuzhen~Cheng$^{\dag}$}
	\IEEEauthorblockA{$^\dag$ School of Computer Science and Technology, Shandong University}
}

\maketitle

\begin{abstract}
Decentralized Storage Networks (DSNs) are emerging as a foundational infrastructure for Web 3.0, offering global peer-to-peer storage. However, a critical vulnerability persists: user privacy during file retrieval remains largely unaddressed, risking the exposure of sensitive information. To overcome this, we introduce PIR-DSN, the first DSN protocol to integrate Private Information Retrieval (PIR) for both single and multi-server settings. Our key innovations include a novel secure mapping method that transforms sparse file identifiers into compact integer indexes, enabling both public verifiability of file operations and efficient private retrieval. Furthermore, PIR-DSN guarantees Byzantine-robust private retrieval through file replication across multiple miners. We implement and rigorously evaluate PIR-DSN against three prominent industrial DSN systems. Experimental results demonstrate that PIR-DSN achieves comparable overhead for file upload and deletion. While PIR inherently introduces an additional computational cost leading to higher retrieval latency, PIR-DSN maintains comparable throughput. These findings underscore PIR-DSN's practical viability for privacy-sensitive applications within DSN environments.
\end{abstract}


\section{Introduction}
\label{sec:introduction}
Decentralized Storage Networks (DSNs) are innovative systems that combine peer-to-peer storage with blockchain technology. What sets them apart from traditional systems is their proof-of-storage mechanism, which cryptographically verifies file storage and ensures data availability even within untrusted networks~\cite{dsnsok}. DSNs pool unused storage from providers worldwide, forming a fault-tolerant decentralized network. DSNs have made strides recently, boasting increased storage capacity~\cite{filedag,ecchain}, improved data security~\cite{bftdsn}, and enhanced privacy protection~\cite{FileDES}. They are quickly becoming essential infrastructure in Web3.0, powering the storage needs of NFTs~\cite{Opensea} and decentralized AI~\cite{Nuklai}.

Privacy is a persistent concern in DSNs, encompassing both data privacy and user privacy. Efforts have been made in data privacy, ensuring file contents are inaccessible to unauthorized parties. Many industrial DSNs, like Sia~\cite{sia}, Storj~\cite{storj}, and Swarm~\cite{swarm}, tackle data privacy by allowing clients to upload encrypted files. Solutions such as FileDES~\cite{FileDES} further enhance this with fine-grained data protection using proxy re-encryption and attribute-based encryption. 
However, user privacy, which aims to prevent the inference of sensitive user characteristics from their file access patterns (e.g., frequency) and related data (e.g., metadata, public descriptions), is still largely unaddressed. In DSNs, clients retrieve files from storage miners using file identifiers (FIDs). This process creates a vulnerability: storage miners can easily observe access frequencies and public descriptions through these FIDs. Observing file access patterns can inadvertently reveal users' habits, personal interests, or cultural affiliations~\cite{Popcorn}. This is especially critical since DSNs often serve as storage for content delivery networks (CDNs) like Saturn~\cite{saturn} and Storacha~\cite{storacha}, and video streaming networks like Livepeer~\cite{livepeer}. 

To bolster user privacy within DSNs, a critical need exists: clients must be able to retrieve publicly available files from storage miners without divulging their access patterns. This requirement finds a strong match in Private Information Retrieval (PIR)~\cite{pir}, which empowers clients to fetch data from servers without revealing what information they are requesting. PIR's mechanism, where clients send query vectors to servers, aligns well with DSN file retrieval, and its multi-server model complements the typical DSN setup where files are stored across multiple miners. Despite this strong conceptual fit, the integration of PIR into DSNs to protect user privacy has not yet been fully explored or implemented. Even with the promising alignment, we foresee challenges that would need to be addressed in such a system, particularly concerning the design of efficient index structures and mitigating malicious attacks. 

\subsubsection{The sparsity of FIDs hinders efficient private information retrieval} Designing PIR-DSN presents a primary challenge -- aligning the FIDs used in DSNs with the indexing schemes required by PIR protocols. Currently, two main approaches for PIR exist: index-based and keyword-based. With index-based PIR, a client privately retrieves a record by its integer index. The challenge here is that DSNs use hash-based FIDs, which are incredibly numerous and sparsely distributed (e.g., \(2^{256}\) for 256-bit hashes). Directly mapping these FIDs to consecutive integers would create an unmanageably large database, leading to an enormous computational burden for index-based PIR, as servers would need to perform $O(n)$ computations for query vectors of size $n$.

Keyword-based PIR is an alternative. ChalametPIR (ACM CCS’24)~\cite{celi2024call}, utilizes binary fuse filters to map \(n\) keywords compact integer indexes ranging from $1$ to $1.08n$. This allows for retrieving a record using its sparsely distributed keyword with three invocations of an underlying index-based PIR, namely FrodoPIR~\cite{frodopir}. However, this method comes with its drawbacks if applied in DSNs. ChalametPIR's mapping method is designed for static datasets, but DSNs are inherently dynamic as miners continually handle upload requests and store new files. Inserting the keywords of new files into a binary fuse filter may introduce collisions, forcing the miner to rebuild the filter with a larger size or new hash functions, resulting in large overhead. Consequently, achieving efficient and secure private retrieval in DSNs critically depends on unifying their index structures with those used in PIR.

\subsubsection{Malicious storage miners can undermine data availability} 
A core goal of DSNs is to ensure high data availability. However, malicious miners pose a big threat, as they are globally distributed and often weakly managed, enabling them to launch attacks that compromise both data integrity and availability~\cite{damgaard2024severe}. This vulnerability distinguishes DSNs from traditional cloud or P2P storage systems. Through their malicious actions, these miners can subtly manipulate clients into storing files within DSNs under their control, ultimately increasing their own profits.

Specifically, malicious miners can execute index manipulation attacks~\cite{prunster2022total,sridhar2023content}. Research~\cite{prunster2022total} has shown that even with modest resources (like a 4-core CPU and 16GB RAM), malicious miners can manipulate FIDs to misdirect retrieval requests to incorrect objects. When we integrate PIR into DSNs, clients not only need to locate the correct miners but also depend on accurate database indexes for successful file retrieval. Since these indexes are generated and maintained by the miners, they are highly vulnerable to manipulation. Consequently, when the client attempts to retrieve a file, it could end up with the wrong file or an empty one. A robust mechanism is essential to prevent these insidious index manipulation attacks throughout the DSNs' lifecycle.

In response, we propose PIR-DSN, a system that enables private file retrieval over DSNs. Our main contributions are summarized as follows:
\begin{itemize}
\item \textbf{Efficient and Verifiable FID-to-Index Mapping:} We introduce an efficient and verifiable method for mapping sparse FIDs to compact integer indexes, a critical step for efficiently employing PIR in dynamic DSNs. Our approach leverages an Asynchronous Cryptographic Accumulator (ACA) to maintain these mappings and generate verifiable proofs for all establishment and deletion operations. This ensures public verifiability of file operations and enables efficient calculation of PIR retrieval results, directly addressing challenges posed by FID sparsity and mitigating malicious index manipulation attacks.

\item \textbf{Private Single-Miner PIR Based DSN (SPIR-DSN):} Our first protocol, SPIR-DSN, empowers clients to privately retrieve files from a single storage miner. By integrating our secure mapping method with a single-server index-based PIR protocol, clients can derive correct file indexes from mapping establishment proofs. This allows clients, in a single invocation of SPIR-DSN, to either successfully retrieve the correct file or detect if a malicious miner attempts index manipulation attacks, all without revealing the specific file being retrieved.

\item \textbf{Robust Multi-Miner PIR Based DSN (MPIR-DSN):} Our second protocol, MPIR-DSN, ensures robust and private file retrieval in multi-miner DSN settings where files are replicated for enhanced security and resilience. We integrate our secure mapping method with a Byzantine fault-tolerant state machine replication (BFT-SMR) protocol to maintain consistent database states across multiple storage miners. This consistency allows us to apply a Byzantine-robust PIR scheme, enabling clients to privately retrieve correct files and detect miners returning incorrect query answers within a single MPIR-DSN invocation.

\item \textbf{Analysis and Evaluation:} We analyze the security and efficiency properties of PIR-DSN from a theoretical perspective. We also implement and rigorously evaluate PIR-DSN against three prominent industrial DSN systems: Filecoin, Sia, and Storj. Experimental results demonstrate that PIR-DSN achieves comparable performance, even with enhanced privacy protection guarantees. 

\end{itemize}

The PIR-DSN is available at: \url{https://github.com/saika2k/PIR-DSN-code.git}

\begin{table*}[!t]
	\begin{center}
	\begin{threeparttable}
		\caption{Comparison of PIR-DSN with related schemes}
            \label{table:features}
            \tabcolsep=0.25cm
			\begin{tabular}{l c c c c c c c c}
				\toprule[1pt]
                & \multicolumn{4}{c}{\textbf{Retrieval}}
                & \multicolumn{4}{c}{\textbf{Index}} \\
                \cmidrule(lr){2-5} \cmidrule(lr){6-9}
                   & \multicolumn{1}{c}{\textbf{\begin{tabular}[c]{@{}c@{}}Servers\\ involved\end{tabular}}}
                   & \multicolumn{1}{c}{\textbf{\begin{tabular}[c]{@{}c@{}}User\\ privacy\end{tabular}}} 
                   & \multicolumn{1}{c}{\textbf{\begin{tabular}[c]{@{}c@{}}Robustness\end{tabular}}}
                   & \multicolumn{1}{c}{\textbf{\begin{tabular}[c]{@{}c@{}}Per-Server\\ Computation\end{tabular}}}
                   & \multicolumn{1}{c}{\textbf{\begin{tabular}[c]{@{}c@{}}Public\\ verifiability\end{tabular}}}
                   & \multicolumn{1}{c}{\textbf{\begin{tabular}[c]{@{}c@{}}Verification\\ complexity\end{tabular}}}
                   & \multicolumn{1}{c}{\textbf{\begin{tabular}[c]{@{}c@{}}Method\end{tabular}}} 
                   & \multicolumn{1}{c}{\textbf{\begin{tabular}[c]{@{}c@{}}Value\\ space\end{tabular}}}\\
                \midrule[0.3pt]
                FileDES\cite{FileDES} & Single & \XSolidBrush & \XSolidBrush & $O(1)$ & \Checkmark & $O(1)$ & KW & $O(2^\lambda)$\\ 

                Sia\cite{sia} & Multiple & \XSolidBrush & \Checkmark & $O(1)$ & \Checkmark & $O(\log n)$ & KW & $O(2^\lambda)$\\
                
				Peer2PIR\cite{mazmudar2024peer2pir} & Single & \Checkmark & \XSolidBrush & $O(n)$ & \XSolidBrush & \XSolidBrush & KW & $O(2^\lambda)$\\
				
				Periscoping~\cite{liu2024periscoping} & Single & \Checkmark & \XSolidBrush & $O(n)$ & \XSolidBrush & \XSolidBrush & Integer & $O(n)$\\

                Tajeddine et al.~\cite{tajeddine2019private} & Multiple & \Checkmark & \Checkmark & $O(n)$ & \XSolidBrush & \XSolidBrush & Integer & $O(n)$ \\

				\textbf{PIR-DSN} & Single or Multiple & \Checkmark & \Checkmark & $O(n)$ & \Checkmark & $O(\log n)$ & KW-to-Integer & $O(n)$\\
    
				\bottomrule[1pt]
			\end{tabular}
		\begin{tablenotes}
            \item[Servers involved] The number of servers participate in a single file retrieval process
            \item[KW] Keyword
            \item[$\lambda$] The security parameter whose value is 256
		\end{tablenotes}
	\end{threeparttable}
	\end{center}
\end{table*}

\section{Related Works}
\label{sec:related:work}
This section first reviews privacy protection schemes in DSNs, then introduces distributed systems employing PIR, and finally analyzes their limitations.

\subsection{Privacy Protection Schemes in DSNs}
Research on privacy protection in DSNs has primarily focused on securing data privacy. Sia~\cite{sia}, Swarm~\cite{swarm} and Storj~\cite{storj} protect data privacy by encrypting files before uploading them to storage miners. Filecoin~\cite{filecoin} and Arweave~\cite{arweave} allow clients to either upload files directly to make them publicly accessible or encrypt them locally before uploading. Recent studies aim to balance data privacy with authorized access to private files. Systems like FileDES~\cite{FileDES} and FileTrust~\cite{filetrust} utilize proxy re-encryption techniques to enable authorized access. In these systems, clients request re-encryption and decryption keys from the uploader, while storage miners re-encrypt the files and return the ciphertext. Clients can then decrypt the files using the decryption key. Wang et al.~\cite{wang2018blockchain} combined attribute-based encryption with DSNs, allowing uploaders to distribute secret keys based on an access policy.


\subsection{Distributed Systems Employing PIR}
PIR allows a client to retrieve data from untrusted servers without revealing which data is being retrieved. This technique has been applied in many distributed systems except DSN. Peer2PIR~\cite{mazmudar2024peer2pir} is designed for the InterPlanetary File System (IPFS). Clients treat FIDs as keywords and use a single-server keyword PIR protocol to retrieve the target peer and file index. Then, a single-server index-based PIR is used to retrieve the file from the target peer. Tajeddine et al.~\cite{tajeddine2019private} proposed a multi-server PIR method to allow clients to privately retrieve distinct MDS codewords from multiple databases, enabling data recovery while maintaining privacy. PIR-Tor~\cite{mittal2011pirtor} addresses the issue of clients retrieving relay information to build communication circuits in the Tor network. It uses single-server PIR to retrieve middle and exit relays from a directory server and multi-server PIR to retrieve exit relays from three guard relays. Periscoping~\cite{liu2024periscoping} focuses on the issue of clients retrieving public keys in Mixnet. It employs multi-server PIR to fetch public keys from each mix in the routing path, embedding the query vectors within onion-encrypted messages to prevent any mix from learning which keys are being retrieved.

\subsection{Limitations of Existing Methods}
While existing methods, as summarized in Table~\ref{table:features}, have advanced privacy in DSNs and similar distributed systems, they still face notable limitations. A primary limitation is their restricted retrieval scope. Current approaches address either single-server or multi-server retrieval, but none offer a solution that handles both. Furthermore, some methods completely fail to protect user privacy. For example, FileDES~\cite{FileDES} and Sia~\cite{sia}'s retrieval processes allow miners to observe which files are requested, thereby learning clients' access patterns. Another critical issue is a lack of robustness. While Peer2PIR~\cite{mazmudar2024peer2pir} and Periscoping~\cite{liu2024periscoping} achieve user privacy, they fall short on robustness. If that miner returns incorrect data, clients can only detect the error, but cannot recover the correct file. Finally, several schemes, including Peer2PIR, Periscoping, and the one proposed by Tajeddine et al.,~\cite{tajeddine2019private} do not consider the public verifiability of indexes. They operate under the assumption that servers are honest-but-curious, meaning servers provide correct indexes but might passively observe queries to infer information, without accounting for potential malicious index manipulation.

\section{Model, Preliminary and Design Goals}
\label{sec:model}
This section outlines the network model that forms the foundation of PIR-DSN. We also introduce essential preliminary concepts and articulate our core design goals.

\subsection{Network Model}
\subsubsection{Participants}
In PIR-DSN, the network comprises three primary participants:
\begin{itemize}
\item \textit{Clients:} devices such as personal computers, laptops, or mobile phones. They pay with cryptocurrency to perform the following operations: 1) upload files to storage miners; 2) delete files from storage miners; 3) retrieve files privately from storage miners.

\item \textit{Storage miners:} entities that provide storage resources. They perform the following tasks to handle client requests and earn profits: 1) store files; 2) delete files; 3) respond to private file retrieval queries. In addition, miners routinely generate storage proofs to demonstrate file integrity and participate in consensus processes to earn mining rewards.

\item \textit{Blockchain:} a decentralized, tamper-resistant bulletin board. It records verifiable proofs to ensure data transparency and public verifiability. 
\end{itemize}

\subsubsection{Adversary} 
\label{subsubsec:Adversary}
Storage miners can be either honest, meticulously adhering to the protocol, or malicious, operating under the control of an adversary. Our security model posits that the adversary controls fewer than one-third of the total miners. Exceeding this threshold would render consensus within PIR-DSN impractical~\cite{lamport2019byzantine}. The adversary's primary goal is to undermine PIR-DSN's data integrity and availability by deviating from the protocol. While existing security mechanisms in DSNs generally address issues like data integrity and content confidentiality, our focus is specifically on new vulnerabilities introduced by integrating PIR. Specifically, the adversary can launch an \textit{index manipulation attack}~\cite{prunster2022total,sridhar2023content}, targeting our novel FID-to-index mapping mechanism. 

In PIR-DSN, this attack can manifest in several ways:
\circlednumber{1} \textit{Upload:} When a file is uploaded, a malicious miner could exploit the system by mapping the new FID to an index already in use. This allows them to trick clients into unknowingly retrieving a different file when querying for the original. Alternatively, the miner could assign the FID to an unoccupied but incorrect index, thereby disrupting the overall index arrangement.
\circlednumber{2} \textit{Deletion:} When a client initiates file deletion, a malicious miner could delete FID mappings the client did not specify, leading to the client receiving an empty response if they later query for that mistakenly removed file. Conversely, the miner might falsely claim to delete a requested FID mapping while leaving it intact, allowing the client to still retrieve a file that should have been deleted.
\circlednumber{3} \textit{Modification (at any time):} A malicious miner could, at any point, arbitrarily alter the index associated with any FID. This means the indexes a client expects for their files might not match what the miner maintains. Consequently, if a client queries for a file whose index has been tampered with, they could receive the wrong file or nothing at all.

Besides, we make the standard assumption that all nodes operate within polynomial time, meaning they do not possess unrestricted computational power. Additionally, we assume that cryptographic primitives, such as hash functions and digital signatures, are secure.

\subsection{Preliminary}
\subsubsection{Private Information Retrieval (PIR)}
PIR enables clients to query a public database at an index $i$ without revealing $i$ to the database servers. PIR protocols are typically categorized as either single-server or multi-server.

Single-server PIR protocols require only one computationally bounded server and are often constructed using homomorphic public-key cryptosystems~\cite{simplepir, sealpir}. A single-server PIR protocol consists of the following algorithms:
$$\mathsf{PIR_{S} = (Query, Answer, Decrypt}).$$ 
Clients use the $\mathsf{Query}$ algorithm to generate and send query vectors, and the $\mathsf{Decrypt}$ algorithm to obtain the queried database record. Servers use the $\mathsf{Answer}$ algorithm to compute query answers and respond to clients. A single-server PIR scheme must satisfy the properties of correctness and privacy. It is correct if an honest client interacting with an honest server always retrieves the correct database record. It is private if a malicious server cannot guess the client’s queried index with probability better than random guessing.

We adopt Multi-server PIR protocols that involve $p>1$ servers holding identical database replicas~\cite{chor1998private,APIR}. They avoid cryptographic hardness assumptions and use lightweight operations such as XOR or function secret sharing. A multi-server PIR protocol consists of the following algorithms:
$$\mathsf{PIR_{M} = (Query, Answer, Reconstruct}).$$

Clients use the $\mathsf{Reconstruct}$ algorithm to obtain the queried database record from multiple query answers. A multi-server PIR protocol must also satisfy the properties of correctness and privacy. It is correct if a client interacting with $p$ honest servers can always recover the correct record. It is private if any coalition of up to $p-1$ malicious servers cannot guess the client’s queried index with probability better than random guessing.

\subsubsection{Asynchronous Cryptographic Accumulator (ACA)}
ACA~\cite{AsyncACC} is a space-efficient data structure that uses cryptographic primitives for compact set representation. It is structured as a Merkle forest, where each element in the set is stored in a leaf. An ACA consists of the following algorithms: 
$$\mathsf{ACA = (Gen, Insert, Delete, VerMem}).$$
The $\mathsf{Gen}$ algorithm initializes an empty ACA. The $\mathsf{Insert}$ and $\mathsf{Delete}$ algorithms insert or remove elements at specific leaves and generate corresponding witnesses $\vec{w}$. The $\mathsf{VerMem}$ algorithm verifies an element’s membership using its witness.


\subsection{Design Goals}
PIR-DSN addresses the attacks by malicious miners and the inefficiency caused by sparse FIDs. Its design should meet four goals:
\begin{itemize}
    \item \textit{Public verifiability:} Throughout a DSN's operation, a verifier can continuously check that mappings are not being manipulated by an adversary performing attacks \circlednumber{1}\circlednumber{2}\circlednumber{3}. With this property, if a violation is found, the system can reject the unauthorized operation.
    \item \textit{Privacy:} When a client retrieves a file from one or multiple miners, an adversary cannot guess the client’s queried index with probability better than random guessing.
    \item \textit{Robustness (in multi-server setting):} When a client retrieves a file from multiple miners, it can both reconstruct the correct file and identify every miner that returns an incorrect query answer.
    \item \textit{Efficiency:} If a miner stores $n$ files, then (1) any verifier can check the correctness of file upload or deletion in $O(\log n)$ complexity, and (2) a client can retrieve a file from one or multiple miners with a single PIR invocation, while the computation complexity of each miner is $O(n)$.
\end{itemize}

\begin{figure*}[!th] 
	\centering 
	\includegraphics[width=1.00\textwidth]{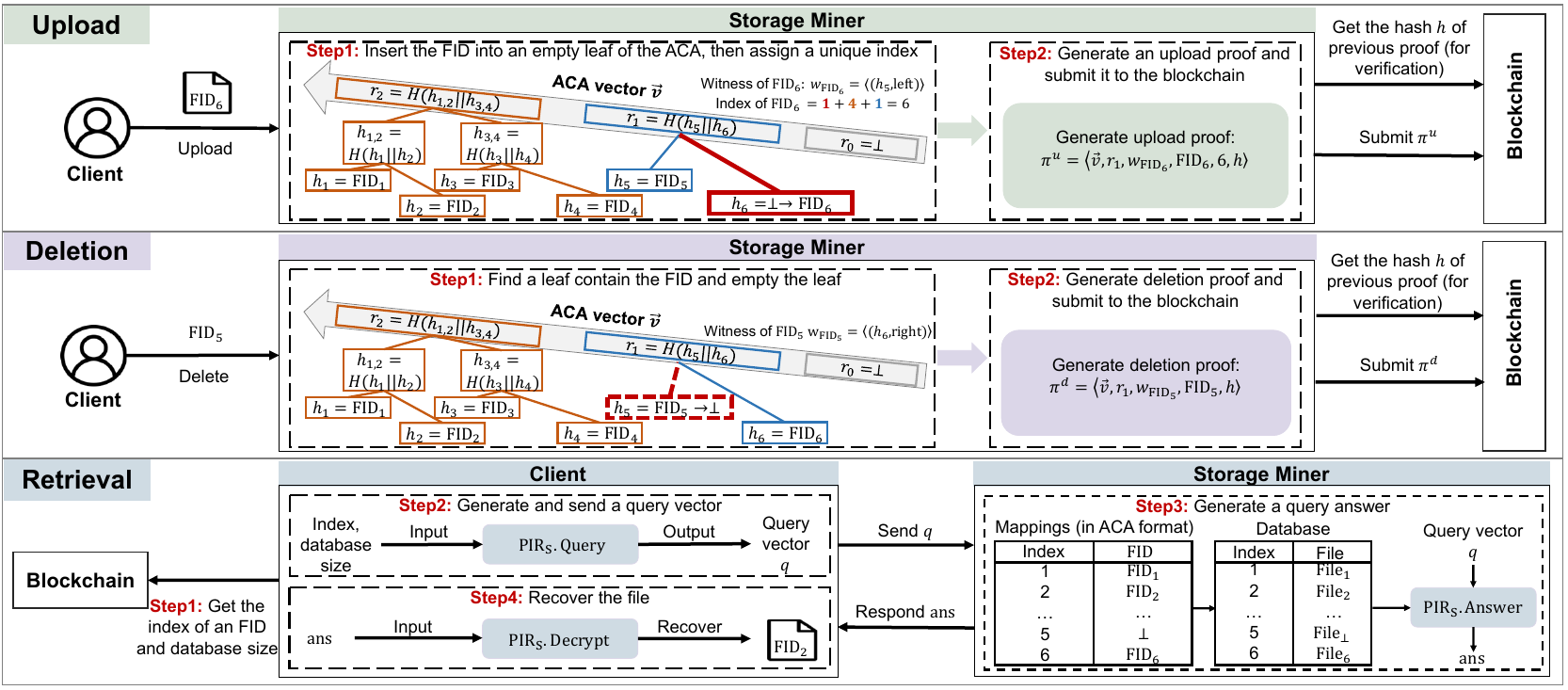} 
	\caption{The diagram of the SPIR-DSN protocol} 
	\label{SPIR protocol}
\end{figure*} 

\section{Single-Miner PIR on DSN}
\label{sec::Single-PIR}
We first propose the Single-Miner Private Information Retrieval Based DSN (SPIR-DSN) protocol, whose diagram is shown in Fig.~\ref{SPIR protocol}. We describe our protocol in terms of the three operations: upload, deletion, and retrieval.

\subsection{File Upload} A key technical challenge in file management, particularly during uploads, is ensuring the verifiability of mappings to detect malicious miners. Such miners might map a new FID to an already-used or incorrect index, thereby disrupting the index arrangement and causing clients to retrieve incorrect files.

Our approach tackles this by using an Asynchronous Cryptographic Accumulator  (ACA) to establish and prove mappings. The ACA consists of a vector $\vec{v}=(r_0,r_1,\dots,r_{m-1})$ and $n$ FIDs (structured as Merkle forest). Each element $r_i$ is either a Merkle root or $\perp$. If $r_i \neq \perp$, each of its $2^i$ leaves either stores an FID or a placeholder $\perp$. If $r_i = \perp$, this represents a Merkle tree where all $2^i$ leaves store $\perp$. With ACA, we can map $n$ FIDs to compact indexes within the range $[1,2,\dots,n]$. Moreover, we leverage ACA to provide a membership witness for each FID-to-index mapping to confirm their presence within the ACA. Then, we generate a proof for each mapping and record it on a bulletin board (i.e., blockchain), enabling the public verification of mapping correctness. Further details of our approach are as follows. 

\subsubsection{Mapping Establishment} 
Our method involves systematically mapping $n$ FIDs by sequentially inserting them into an ACA, with each FID assigned an integer index within the compact range $[1, 2, \dots, n]$. The ACA vector is initialized using $\mathsf{ACA.Gen}$ and defined as $\vec{v} = (r_0)$, where $r_0 = \perp$ (representing vacancy). The state of an ACA is represented by its vector $\vec{v} = (r_0, r_1, \dots, r_{m-1})$, where $m$ is the length of the vector. Two scenarios arise during the insertion of an FID:
\begin{enumerate}
    \item Full ACA: All leaves in the ACA are occupied by FIDs. In this scenario, it is necessary to extend the ACA to accommodate additional insertions.
    \item Partial ACA: One or more leaves contain $\perp$, indicating vacant leaves. In this case, the FID is inserted into the first available vacant leaf.
\end{enumerate}


Specifically, the ACA vector $\vec{v}$ is traversed from $r_{m-1}$ to $r_0$ to locate the first $r_i$ that satisfies one of the following conditions: (1) $r_i = \perp$; or (2) $r_i \neq \perp$ and the Merkle tree associated with $r_i$ contains leaves storing $\perp$. If no such $r_i$ is found (Full ACA),  a new Merkle tree $r_m$ is created to provide additional index space. We then invoke the $\mathsf{ACA.Insert}$ algorithm and the newly inserted FID is then merged with the $2^m-1$ leaves of the Merkle trees spanning $r_0$ to $r_{m-1}$, forming a complete Merkle tree with $2^m$ leaves. The root of this new tree is stored in $r_m$, and the FID is placed in the rightmost leaf. Subsequently, the values from $r_0$ to $r_{m-1}$, as well as all merged leaves, are set to $\perp$. If an appropriate $r_i$ is found (Partial ACA), the FID is inserted into the leftmost leaf if $r_i = \perp$. If $r_i \neq \perp$, a pre-order depth-first search (DFS) is conducted to traverse the leaves of the associated Merkle tree, and the FID is inserted into the first leaf that stores $\perp$ by invoking the $\mathsf{ACA.Insert}$ algorithm.

The index assignment process occurs after the FID is correctly inserted. Consider the case where a new FID is inserted into a Merkle tree with root $r_k$. The index is determined using three parameters: $\vec{v}$, $r_k$, and the FID membership witness $\vec{w}$. The first step involves using $\vec{v}$ and $r_k$ to compute the number of leaves preceding the root of the Merkle tree $r_k$. In the next step, the FID membership witness $\vec{w}$ is utilized to identify the position of the newly inserted FID and to compute the number of leaves inserted prior to the FID within the same tree. The vector $\vec{w}$ represents the Merkle path of the FID, where each element consists of a hash and a position indicator (either $\mathsf{left}$ or $\mathsf{right}$).
To determine the index of the new FID, we initialize a counter $\mathsf{index} = 1$. We then iterate through the ACA vector $\vec{v}$ from $r_{k+1}$ to $r_{m-1}$. For each $r_i \neq \perp$, we increment the index by $2^i$. Subsequently, we traverse the Merkle tree $\vec{w}$ from $w_0$ to $w_{k-1}$. For each node $w_j$ that is a left child, we increment the index by $2^j$. For example, in Fig.~\ref{SPIR protocol}, the index of the newly inserted $\mathsf{FID}_6$ is calculated as follows. We start with $\mathsf{index} = 1$. Since $v_2$ is not empty, we add $2^2 = 4$ to the index. As $w_0=\mathsf{FID}_5$ is a left child, we add $2^0 = 1$ to the index. The final index is $1 + 4 + 1 = 6$. 

\subsubsection{Upload Verification} 
After the mapping is established, the storage miner responsible for maintaining the FID submits an upload proof to the blockchain for verification of the mapping. The upload proof is represented as $\pi^{u} = \langle \vec{v}, r_k, \vec{w}, \mathsf{FID}, \mathsf{index}, h \rangle$, where $h$ denotes the hash of the previous state proof generated by the storage miner during the last file upload or deletion. The verification of $\pi^{u}$ involves three distinct steps:

\begin{enumerate}
    \item \textbf{FID existence}: The presence of the FID in the ACA is confirmed by utilizing $\vec{v}$, $r_k$, $\vec{w}$, and $\mathsf{FID}$. Specifically, we invoke the $\mathsf{ACA.VerMer}$ algorithm to verify that $r_k$ corresponds to the $k$-th element of $\vec{v}$ and can be recalculated using $\vec{w}$ and $\mathsf{FID}$. 
    \item \textbf{Correct index position}: The accuracy of the index assignment is verified by checking $\vec{v}$, $r_k$, $\vec{w}$, and $\mathsf{index}$. This involves re-running the index assignment process with $\vec{v}$, $r_k$, and $\vec{w}$ to ensure the result matches $\mathsf{index}$. 
    \item \textbf{Conflict-free index}: To ensure that the newly inserted FID does not result in an index conflict, we use $\vec{v}$, $\vec{w}$, and $h$ to compare the current and previous ACA states. Let $\vec{v}^{t} = (r^{t}_0, r^{t}_1, \dots, r^{t}_{m-1})$ represent the ACA vector of the current state, and $\vec{v}^{(t-1)} = (r^{t-1}_0, r^{t-1}_1, \dots, r^{t-1}_{l-1})$ be the previous state, obtained from $h$ stored in the blockchain. Here, $r^{t}_k$ is the root of the Merkle tree containing the FID, and $r^{t-1}_k$ is the root from the same position in $\vec{v}^{(t-1)}$. By treating the ACA as a state machine with chain-structured transitions, we can efficiently determine if the operation can transition $\vec{v}^{(t-1)}$ to $\vec{v}^{(t)}$. 
\end{enumerate}




\subsection{File Deletion} 
Similar to the file upload operation, the primary challenge in file deletion involves ensuring the verifiability of the mapping removal. Malicious miners may delete FID mappings the client did not specify or falsely claim to delete a requested FID mapping while leaving it intact, causing clients either unable to retrieve expected files or still able to access files that should have been deleted.


To resolve this issue, we are adapting our file upload verification process to secure file deletion. Specifically, we remove an FID from the ACA, generate a proof of this deletion, and record this proof on a blockchain. This method ensures that storage miners in our DSN can publicly verify the integrity of the mapping deletion. With the ability to accurately insert or remove FIDs, we can effectively manage FID updates.

\subsubsection{Mapping Deletion} To remove an FID, we first identify the $\mathsf{leaf}$ containing the FID and the corresponding Merkle tree $r_k$ that includes this $\mathsf{leaf}$. The deletion process is determined by whether $r_k$ contains other FIDs, and can be categorized into two scenarios: (1) Case 1: If $r_k$ contains additional FIDs, we invoke the $\mathsf{ACA.Delete}$ algorithm to construct the membership witness $\vec{w}$ for the FID and update the value stored at $\mathsf{leaf}$ to $\perp$. Subsequently, the root of the Merkle tree is recalculated, and the updated value is stored in $r_k$. (2) Case 2: If $r_k$ contains no other FIDs, we also invoke the $\mathsf{ACA.Delete}$ algorithm to generate the membership witness $\vec{w}$ ($\vec{w}=\perp$ in this case) for the FID and set the value at $\mathsf{leaf}$ to $\perp$. Since all leaves in $r_k$ now hold the value $\perp$, the entire Merkle tree $r_k$ is set to $\perp$.

\subsubsection{Deletion Verification} After mapping deletion, the storage miner submits a proof $\pi^{d}$ = $\langle$$\vec{v}$, $r_k$, $\vec{w}$, $\mathsf{FID}$, $h$$\rangle$ to the blockchain for deletion verification purposes. The verification of $\pi^{d}$ involves two steps: 

\begin{enumerate}
    \item \textbf{FID Integrity in ACA State:} We utilize the vectors $\vec{v}$, $\vec{w}$, the function $\mathsf{FID}$, and the value $h$ to confirm that the FID exists in the previous state of the ACA and that the mappings of other FIDs remain unaffected by the deletion. Initially, we retrieve the prior ACA state $\vec{v}^{(t-1)}$ corresponding to $h$ from the blockchain. Subsequently, we invoke the $\mathsf{ACA.VerMer}$ algorithm to check if the combination of $\vec{w}$ and $\mathsf{FID}$ can recalculate $r^{t-1}_k$, which verifies the presence of the FID in the previous ACA state. Finally, we examine whether all elements in both the current and previous state vectors are identical, except for the $k$-th element, confirming that the mappings of FIDs in other Merkle trees have not been modified. 
    \item \textbf{Correct Execution of Mapping Deletion:} To validate that the miner has properly executed the mapping deletion, we utilize the vectors $\vec{v}$, $r_k$, $\vec{w}$, and $\mathsf{FID}$. We first assess whether $r_k = \perp$. If $r_k \neq \perp$, this indicates that $r_k$ contained other FIDs before the deletion. We then verify that $r_k$ corresponds to the $k$-th element of $\vec{v}$ and can be recalculated using $\vec{w}$ and $\perp$. If $r_k = \perp$, it suggests that $r_k$ only contained the FID before deletion. In this scenario, we verify if $\vec{w}=\perp$. 
\end{enumerate}


\subsection{File Retrieval} 
In SPIR-DSN, a client retrieves files from a storage miner using a single-server PIR protocol. Specifically, the client first obtains the index of an FID and the database size from the upload proof recorded by the blockchain. Next, the client uses the $\mathsf{PIR_{S}.Query}$ algorithm to construct a query vector and send it to the miner. The miner constructs its database by placing the file corresponding to the FID with index $i$ as the $i$-th entry and filling the remaining entries with blank files. Upon receiving the query vector, the miner uses $\mathsf{PIR_{S}.Answer}$ algorithm to generate the query answer and transmit it to the client. Finally, the client uses $\mathsf{PIR_{S}.Decrypt}$ to recover the file and verify its integrity using the FID.

\section{Multi-Miner PIR on DSN}
\label{sec::Multi-PIR}
The Multi-Miner PIR Based DSN (MPIR-DSN) protocol lets clients fetch files from several storage miners at the same time. While you could retrieve a file from a single miner using the SPIR-DSN protocol, MPIR-DSN makes it possible to retrieve from multiple miners concurrently. This is a step up because, unlike single-server PIR, which often needs heavy-duty homomorphic multiplications, multi-server PIR uses simple operations like bitwise XOR. This boosts retrieval efficiency in these kinds of storage setups.

MPIR-DSN and SPIR-DSN cater to different application scenarios. In MPIR-DSN, all storage miners (specifically, $N=3f+1$ of them) hold identical copies of the database and manage all file operations. This full replication means that storage overhead grows linearly with the number of miners. Because of this, MPIR-DSN is particularly well-suited for smaller, latency-sensitive DSNs, such as those found in enterprise partnerships or regional CDN clusters.

\subsection{Database State Replication} A primary challenge in MPIR-DSN is maintaining a consistent database state, which is essential for accurate file retrieval using multi-server PIR. The database state is dynamic, being continuously modified through file upload and deletion requests. However, malicious miners may launch index manipulation attacks, resulting in inconsistent database states across miners. This inconsistency can lead to incorrect file retrieval by clients. 
We address this challenge using a Byzantine Fault-Tolerant State Machine Replication (BFT-SMR) protocol. Specifically, MPIR-DSN instantiates the HotStuff~\cite{hotstuff} protocol and integrates our ACA-based proof generation and verification process, which we identify with \uline{an underline} when describing the following protocol. Our protocol enables miners to reach consensus on the order of client requests and to apply these requests in an agreed order.

\begin{enumerate}
    \item \textit{Order the requests and generate proofs:} 
    Miners take turns acting as the primary miner. The primary miner orders the received client requests, \uline{invokes the $\mathsf{ACA.Insert}$ algorithm to generate upload proofs ($\pi^u$) for upload request, and invokes the $\mathsf{ACA.Delete}$ algorithm to generate deletion proofs ($\pi^d$) for deletion request.} These proofs are then packaged into a block and broadcast to the other miners.
    \item \textit{Verify proofs of the requests:} Upon receiving a block, \uline{each miner verifies all $\pi^u$ and $\pi^d$ in the block using the appropriate verification processes, ensuring the correctness of the proofs and the validity of the associated client requests.} If the block is valid, the miner broadcasts a prepare message.
    \item \textit{Commit and finalize the order:} Once a miner collects at least $2N/3$ prepare messages, it broadcasts a commit message. Upon receiving at least $2N/3$ commit messages, the miner finalizes the block and processes the requests in the agreed order.
\end{enumerate}
If any step fails to complete within a predefined timeout, the system invokes the view-change protocol to select a new primary miner, thereby ensuring liveness.

\subsection{Byzantine-Robust File Retrieval} A secondary challenge in MPIR-DSN is ensuring robust file retrieval. In multi-server PIR, a client distributes queries to multiple miners and aggregates their responses to reconstruct the requested file. However, malicious miners may arbitrarily alter the index associated with any FID to return incorrect responses to clients. Without robustness, the client may be unable to retrieve the correct file and cannot detect miners providing incorrect responses.

To address this issue, we integrate a Byzantine-Robust Private Information Retrieval (BR-PIR) scheme~\cite{goldberg2007improving}. BR-PIR is a specialized multi-server PIR protocol that not only preserves user privacy but also tolerates a certain number of malicious miners returning incorrect answers. As a result, clients can still retrieve the correct file while identifying the faulty miners.
In MPIR-DSN, the client first obtains the file index and database size from the upload proofs, then invokes the $\mathsf{PIR_M.Query}$ algorithm of the BR-PIR protocol to generate $N$ query vectors, which are sent to the $N$ miners. Each miner computes a response using the $\mathsf{PIR_M.Answer}$ algorithm. Upon receiving the responses (or after a timeout), the client applies the $\mathsf{PIR_M.Reconstruct}$ algorithm to recover the file and identify miners that returned incorrect results.



\begin{figure*}[!t]
\centering
\subfigure[File upload]{
\label{File Upload}
\includegraphics[width=0.23\linewidth]{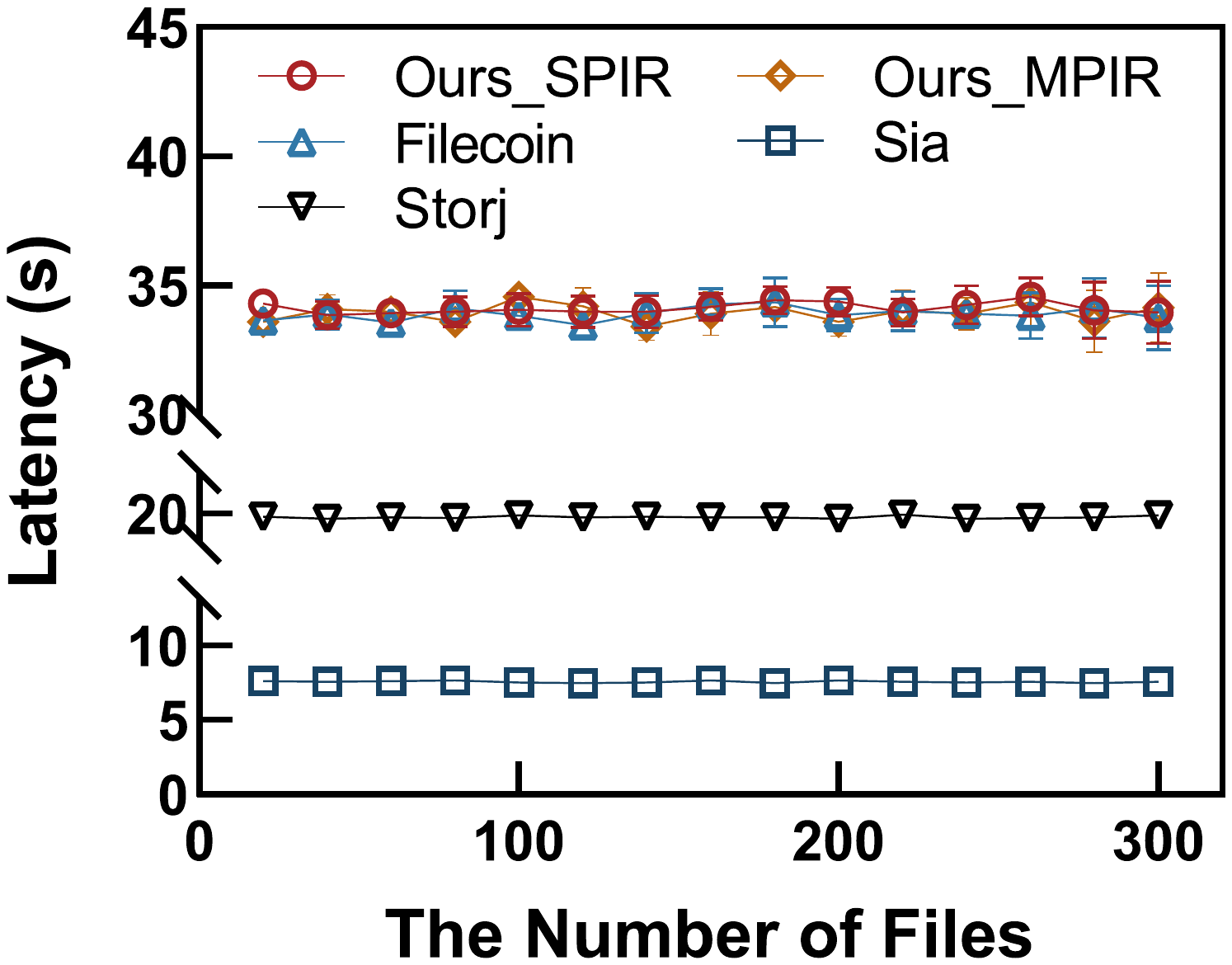}}
\subfigure[File deletion]{
\label{File Deletion}
\includegraphics[width=0.23\linewidth]{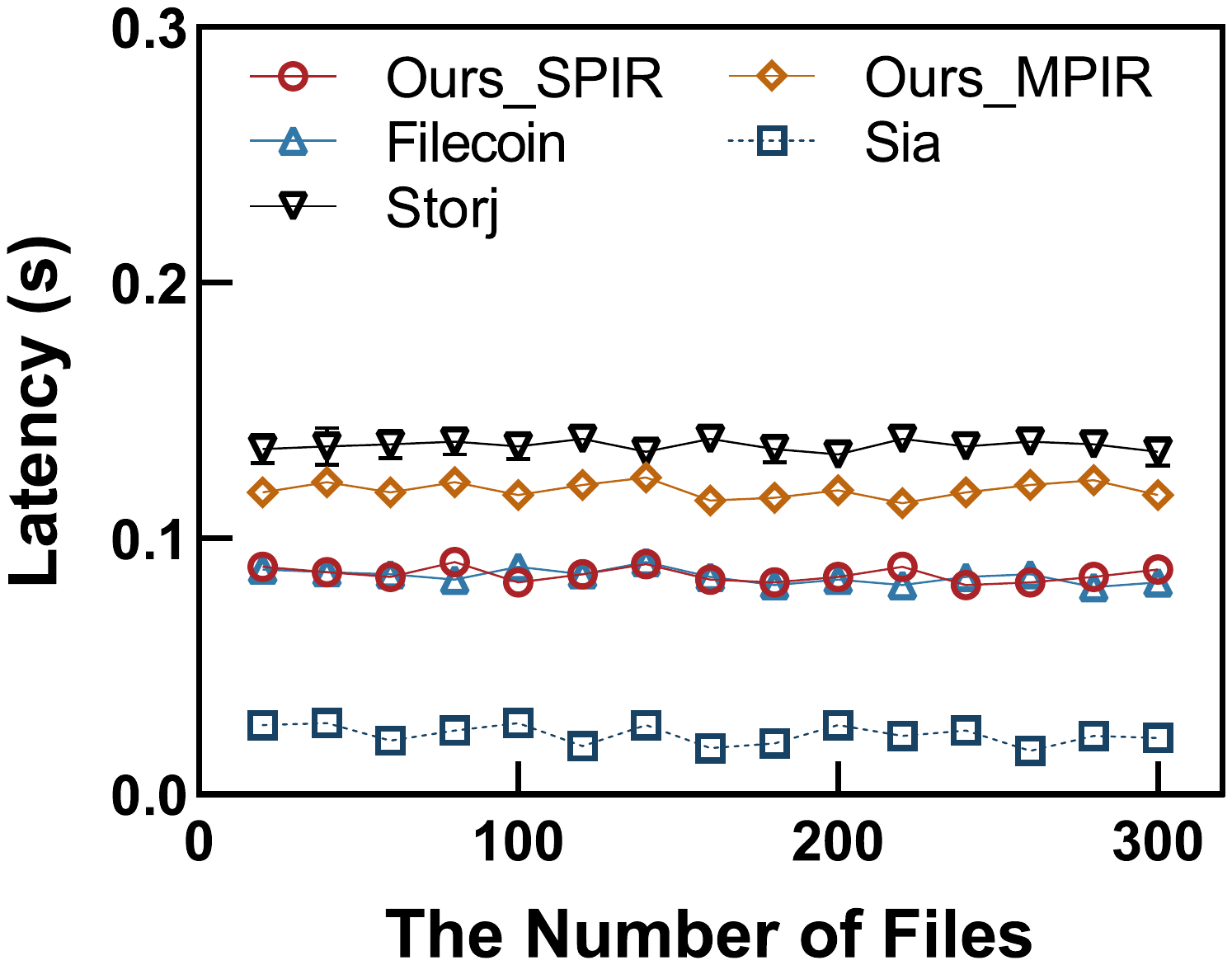}}
\subfigure[Proof verification]{
\label{Proof Verification}
\includegraphics[width=0.23\linewidth]{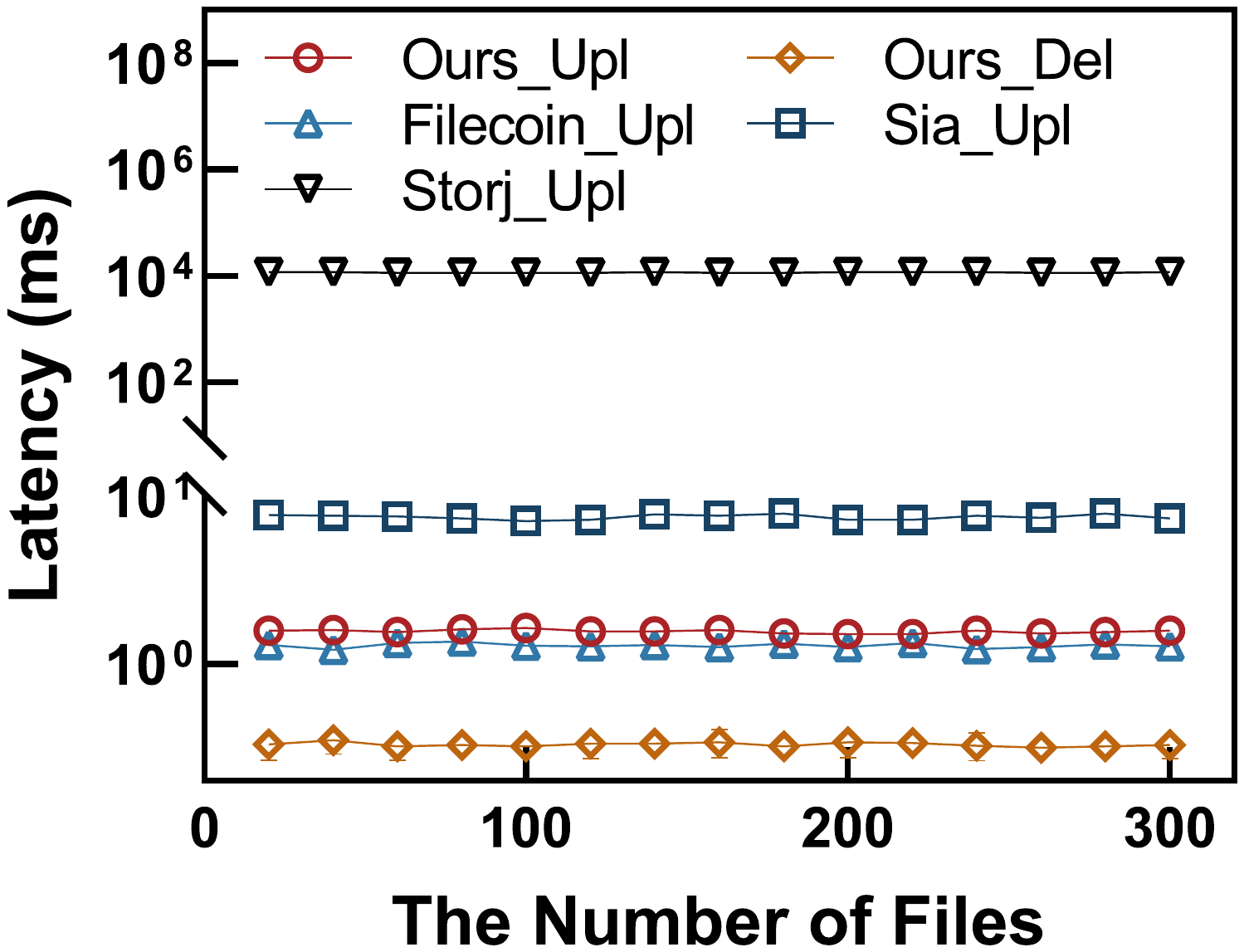}}
\subfigure[File retrieval]{
\label{File Retrieval}
\includegraphics[width=0.23\linewidth]{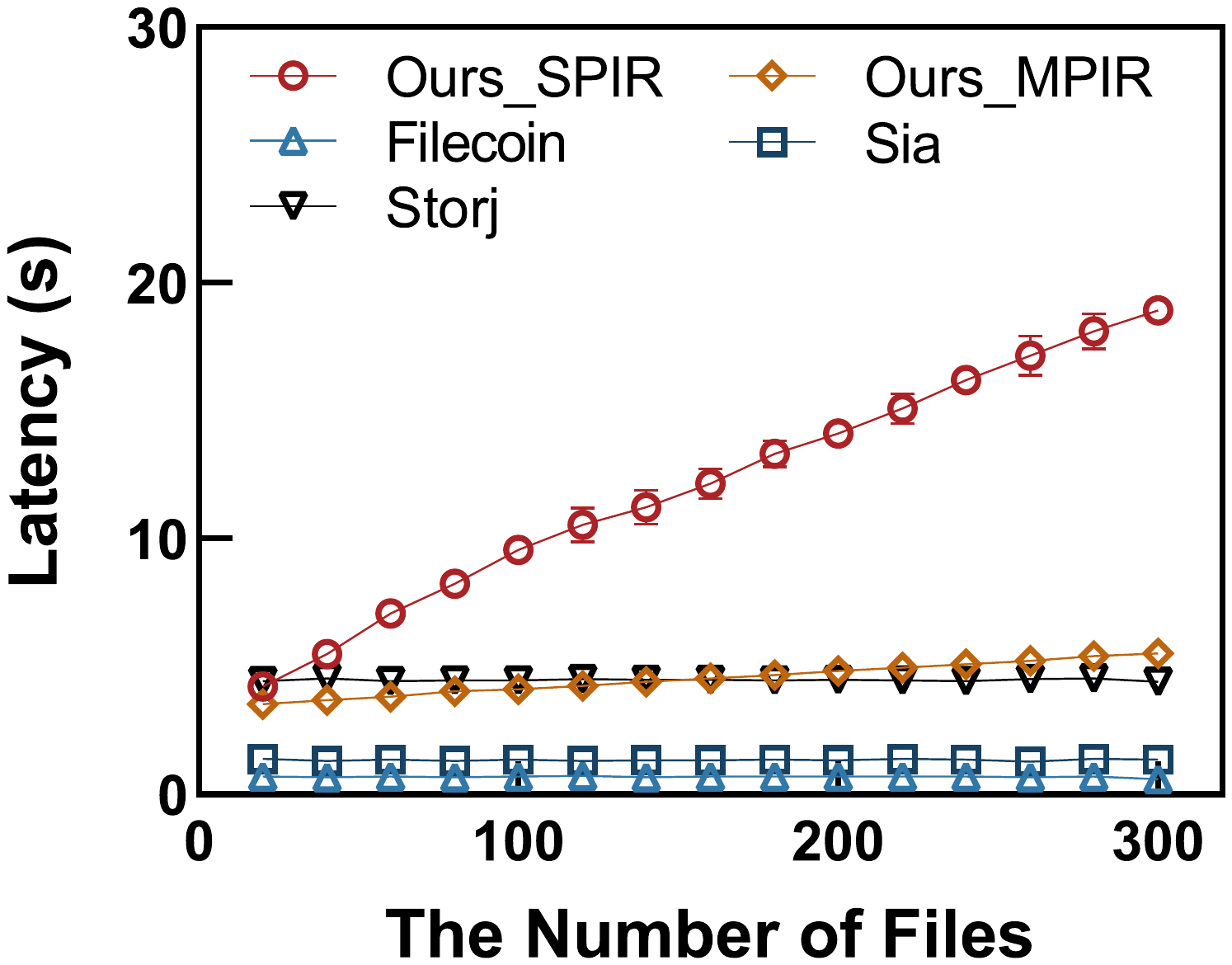}}
\caption{Latency of four basic operations}
\label{File Operation}
\end{figure*}

\section{Analysis}
In this section, we analyze the security of the PIR-DSN protocol, showing it satisfies public verifiability, privacy, robustness, and efficiency as defined in our design goals.

\begin{theorem}[]
    Assuming a storage miner stores $n$ files, PIR-DSN achieves advantages in terms of both security and efficiency. (1) Security: PIR-DSN satisfies public verifiability, privacy, and robustness as defined in our design goals. (2) efficiency: it can verify a file upload and deletion both in $O(\log n)$ complexity and retrieve a file from one or multiple miners with a single PIR invocation, requiring each miner to perform $O(n)$ computation.
\label{theorem:security}
\end{theorem}
\begin{proof}
    For public verifiability, during upload verification, a verifier recomputes the index for the new FID and checks that this position in the previous ACA vector is empty. If a malicious miner maps the FID to (1) an index outside the ACA range or (2) an occupied index, the inconsistency is detected and the upload proof $\pi^{u}$ is rejected. 
    During deletion verification, assume a client requests to delete an FID mapping in $r_k$. A verifier checks (1) the witness in the deletion proof $\pi^d$ can be used to compute the $k$-th element in both the previous and current ACA vector, and (2) all elements except the $k$-th in the current and previous ACA vector are the same. Thus, a miner should only delete the request FID in $r_k$ to pass these two verifications. 

    For privacy, in SPIR-DSN, a client submits one query using $\mathsf{PIR_S.Query}$. The single-server PIR protocol ensures that an adversary gains no advantage over random guessing. In MPIR-DSN, the client sends $N$ queries using $\mathsf{PIR_M.Query}$, one to each of $N=3f+1$ miners. The multi-server PIR protocol guarantees that even if up to $N-1$ miners collude, they cannot infer the queried index. Adversaries in our model can control at most $(N-1)/3$ miners, which is below the threshold needed to compromise privacy.

    To achieve robustness in MPIR-DSN, we utilize Goldberg's BR-PIR protocol. Therefore, we first introduce a lemma, formally proven in his work~\cite{goldberg2007improving}, that underpins its robust properties.
    \begin{lemma}
    Assume at most $c$ servers collude in a multi-server PIR protocol. Using BR-PIR, if a client receives $k$ responses and $h > \sqrt{kc}$ of them are honest, the BR-PIR protocol can return $(G_h, B_\beta)$, where $G_h$ is the set of honest servers and $B_\beta$ is the correct block.
    \label{lemma:robustness}
    \end{lemma}
    In MPIR-DSN, there are $N=3f+1$ miners and at most $f$ corrupt miners who may return incorrect responses. According to Lemma~\ref{lemma:robustness}, in the worst case, there are $c=(N-1)/3$ colluding miners, clients receive $k=N$ responses and $h=(2N+1)/3$ of them are honest. Correctness under the BR-PIR protocol requires $(2N+1)/3 > \sqrt{N(N-1)/3}$. Since this inequality holds for all $N \geq 1$, clients receive enough honest responses to get $(G_h, B_\beta)$. The complement of $G_h$ is the set of misbehaving miners, and $B_\beta$ is the correct file.  

    For efficiency, when a storage miner stores $n$ files, the ACA vector has length $O(\log n)$, and its highest Merkle tree contains $n/2$ leaves. The longest witness, therefore, contains $O(\log n)$ hashes.
    During upload verification, the verifier computes the value of the corresponding position in the current and previous ACA vector using the witness. Both steps require $O(\log n)$ computations. Traversing the ACA vector and the witness in the second step also needs $O(\log n)$ computations. Thus, the total complexity of upload verification is $O(\log n)$.
    For deletion verification, the verifier first computes the value of the corresponding position in the previous ACA vector using the witness and compares all elements in the current and previous ACA vectors. This step requires $O(\log n)$ computations. The second step, computing the value of the corresponding position in the current ACA vector using the witness, also needs $O(\log n)$ computations. Therefore, the complexity of deletion verification is $O(\log n)$.    
    Both SPIR-DSN and MPIR-DSN invoke the underlying PIR protocol once for retrieval. Since an ACA vector with $O(\log n)$ length corresponds to $O(n)$ leaves, the query vector and the database seen by each miner both have a size of $O(n)$. This results in an $O(n)$ computational complexity per miner for retrieval.
\end{proof}

\section{Evaluation}
\label{sec:evaluation}
\subsection{Implementation and Experiment Setup}
\subsubsection{Implementation}
We implement PIR-DSN on top of Filecoin, a widely used DSN developed in Golang. On the client side, we modify the retrieval module so clients can obtain the target file’s index from the blockchain and use the PIR protocol to generate queries and reconstruct files. We use SealPIR \footnote{https://github.com/microsoft/SealPIR} as our single-server PIR scheme and percy++ \footnote{https://sourceforge.net/projects/percy/files/} for BR-PIR. On the miner side, we modify the storage module to support mapping establishment and deletion, and the retrieval module to handle client queries. We also modify the consensus module to integrate database state replication with HotStuff \footnote{https://github.com/relab/hotstuff}. 

\subsubsection{Experiment Setup}
We compare our system with three widely used industrial DSN systems: Filecoin, Sia, and Storj. The evaluation consists of two parts. In Section~\ref{sec::Single_Operation}, we measure the latency of single basic operations, including upload, deletion, proof verification, and retrieval. We deploy the system on a DELL PowerEdge R740 server. Each DSN is configured with four storage miners and one client. The dataset comprises 300 randomly generated text files, each 7.5 MB in size. To ensure fairness, all miners in each DSN store one copy of each file. 

In Section~\ref{sec::Overall_Performance}, we evaluate the overall throughput and latency of file operations. The four DSNs are deployed on a wide-area network (WAN) consisting of 125 SA5.MEDIUM8 instances. Each instance has a 2-core CPU, 8 GB of memory, a 50 GB SSD, and 100 Mbps bandwidth. Of the 125 instances, 100 as storage miners and the remaining 25 as clients. When testing the MPIR-DSN protocol, the 100 storage miners are divided into 25 subnets, each consisting of four miners. The dataset remains the same as in the first part. The number of operations sent by each client varies from one to five per second, with each request targeting either four miners or a subnet in MPIR-DSN. During file retrieval, each storage miner in the four DSNs stores 50 files.

\subsection{Latency of Single Operations}
\label{sec::Single_Operation}


\subsubsection{File Upload} As depicted in Fig.~\ref{File Upload}, the upload latency of SPIR-DSN and MPIR-DSN closely matches those of Filecoin, indicating that our system introduces minimal overhead. These overheads are negligible compared to the 33 seconds needed for generating storage proofs in our system and Filecoin. 
The upload time of Sia and Storj is about 7.53 seconds and 18.96 seconds, respectively. As Sia implements its Merkle tree with 64-byte leaves (two times larger than PIR-DSN and Filecoin)~\cite{sia}, Storj generates storage proofs by directly auditing file blocks~\cite{storj}, these systems improve the efficiency of proof generation but have weak security strength. 

\subsubsection{File Deletion} As shown in Fig.~\ref{File Deletion}, the deletion latency of SPIR-DSN is slightly larger than Filecoin, primarily due to the overhead introduced by mapping deletion. MPIR-DSN experiences a deletion latency longer than SPIR-DSN and Filecoin because of database state replication. Overall, our system introduces minimal additional overhead during the file deletion process. 
The deletion latency of Sia is about 0.023 seconds and is marked with a dotted line. This is because its deletion process does not involve actual file deletion by storage miners; instead, it only removes the FID from the client’s side. Storj relies on a satellite node to relay deletion requests to storage miners, taking approximately 0.136 seconds to delete a file.

\subsubsection{Proof Verification} Fig.~\ref{Proof Verification} illustrates the proof verification latency during file upload for all four systems, as well as during file deletion in our system. Our system achieves the shortest verification latency during file deletion, indicating minimal overhead for deletion verification. The upload verification latency of our system is slightly longer than Filecoin, as both storage proofs and mapping establishment proofs need to be verified. The overhead of mapping establishment verification remains minimal.
Sia and Storj partition files into multiple 256 KB sectors, requiring multiple proofs to verify a single file. The verification latency of Sia and Storj is about 7.59 milliseconds and 11.63 seconds, respectively. 

\subsubsection{File Retrieval} As depicted in Fig.~\ref{File Retrieval}, the retrieval latency of SPIR-DSN increases linearly. This is because storage miners process the entire database when executing the $\mathsf{Answer}$ algorithm of PIR protocols, resulting in computation proportional to the number of stored files. The MPIR-DSN has low retrieval latency becasue of distributed services and getting rid of the computationally intensive homomorphic multiplication in SealPIR. 
Sia takes about 1.35 seconds to retrieve a file. Storj requires about 4.47 seconds to retrieve a file. This discrepancy arises from Storj's use of AES-GCM encryption with a block size of 1 KB~\cite{storj}, which delays decryption and increases overhead. 


\subsection{Overall Throughput and Latency}
\label{sec::Overall_Performance}
These experiments evaluate the overall throughput and latency performance of PIR-DSN under various conditions.

\subsubsection{Upload Throughput and Latency} 
As depicted in Fig.~\ref{Upload Throughput} and Fig.~\ref{Upload Latency}, the throughput of our system is slightly lower and the latency is slightly higher than Filecoin, indicating minimal overhead during mapping establishment. However, due to the computationally intensive sealing process~\cite{filecoin}, the throughput of our system and Filecoin converge when a storage miner processes five files simultaneously. 
Sia exhibits linear throughput growth and stable latency, demonstrating good scalability. Storj’s throughput remains constant while its latency increases linearly.

\begin{figure}[!htbp]
\centering
\subfigure[Throughput]{
\label{Upload Throughput}
\includegraphics[width=0.46\linewidth]{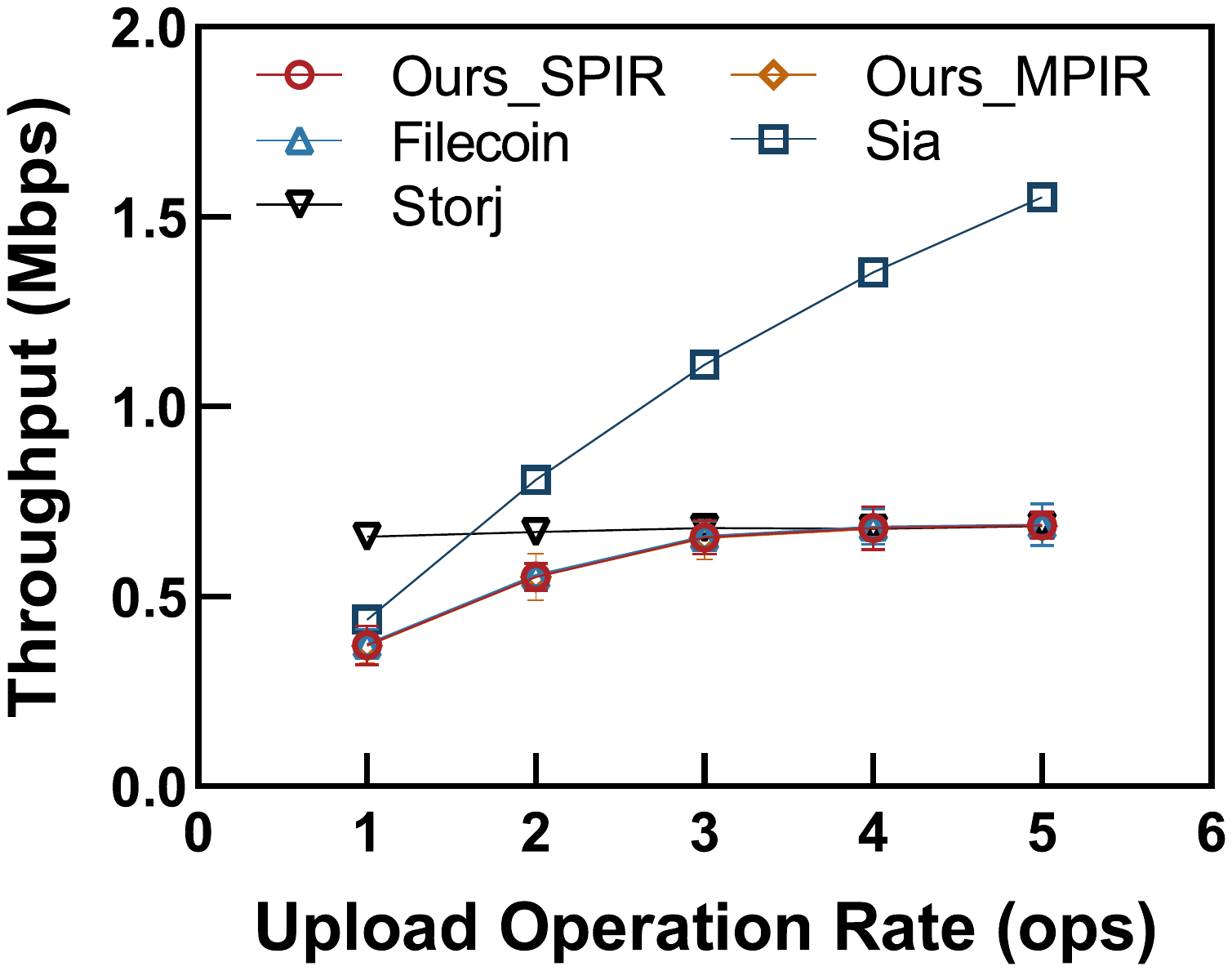}}
\subfigure[Latency]{
\label{Upload Latency}
\includegraphics[width=0.46\linewidth]{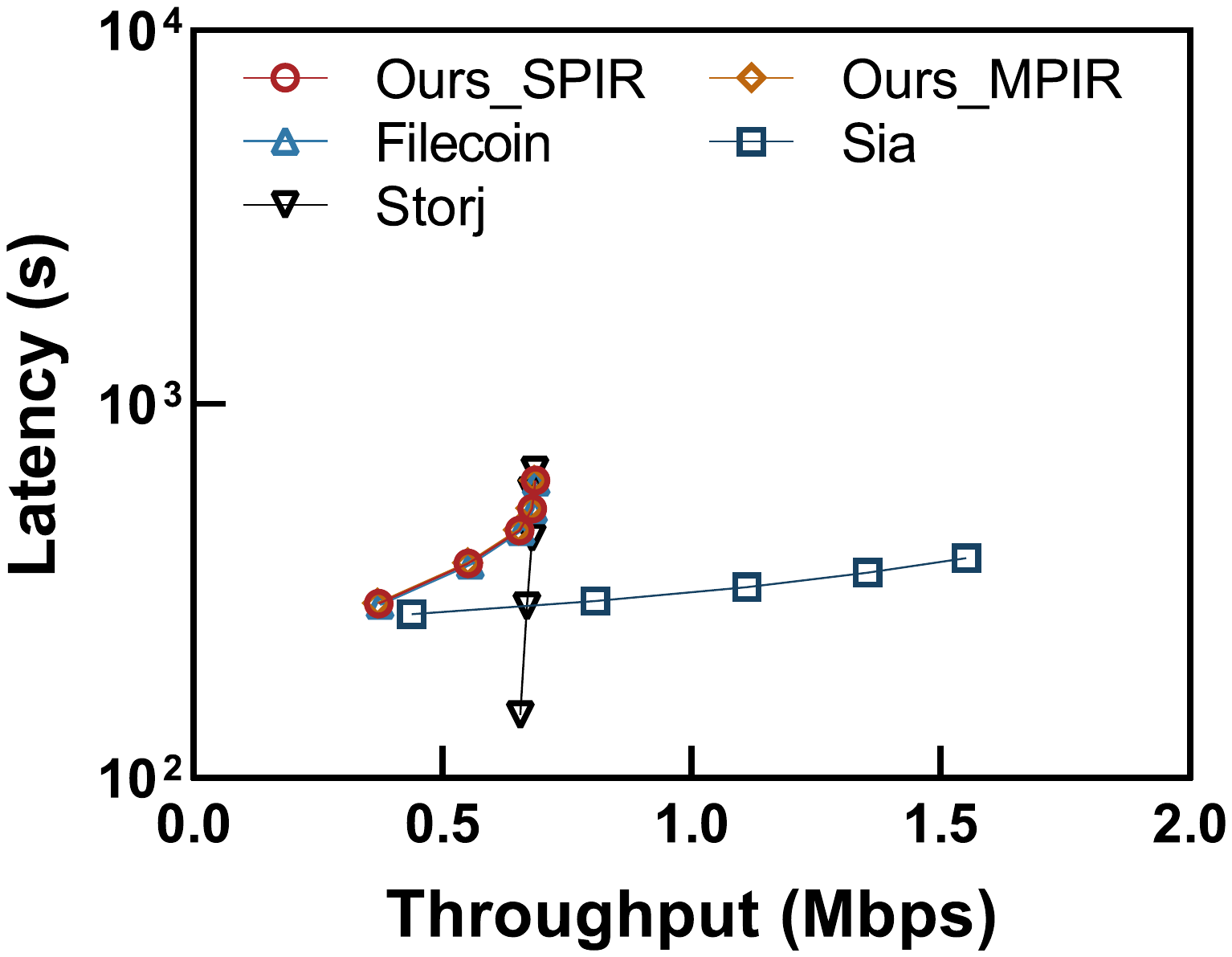}}
\caption{The throughput and latency of file upload operation}
\label{Multi Upload}
\end{figure}

\subsubsection{Deletion Throughput and Latency} As depicted in Fig.~\ref{Deletion Throughput} and Fig.~\ref{Deletion Latency}, the throughput of our system scales linearly and is slightly lower than Filecoin, while the latency remains nearly constant and slightly higher than Filecoin. This suggests that the file deletion process of our system has good scalability, and the mapping deletion process introduces minimal overhead. 
Storj relies on a satellite node to manage deletion requests, resulting in constant throughput and linearly increasing latency. Sia is excluded from this experiment since it does not perform actual file deletions at storage miners. 

\begin{figure}[!htbp]
\centering
\subfigure[Throughput]{
\label{Deletion Throughput}
\includegraphics[width=0.46\linewidth]{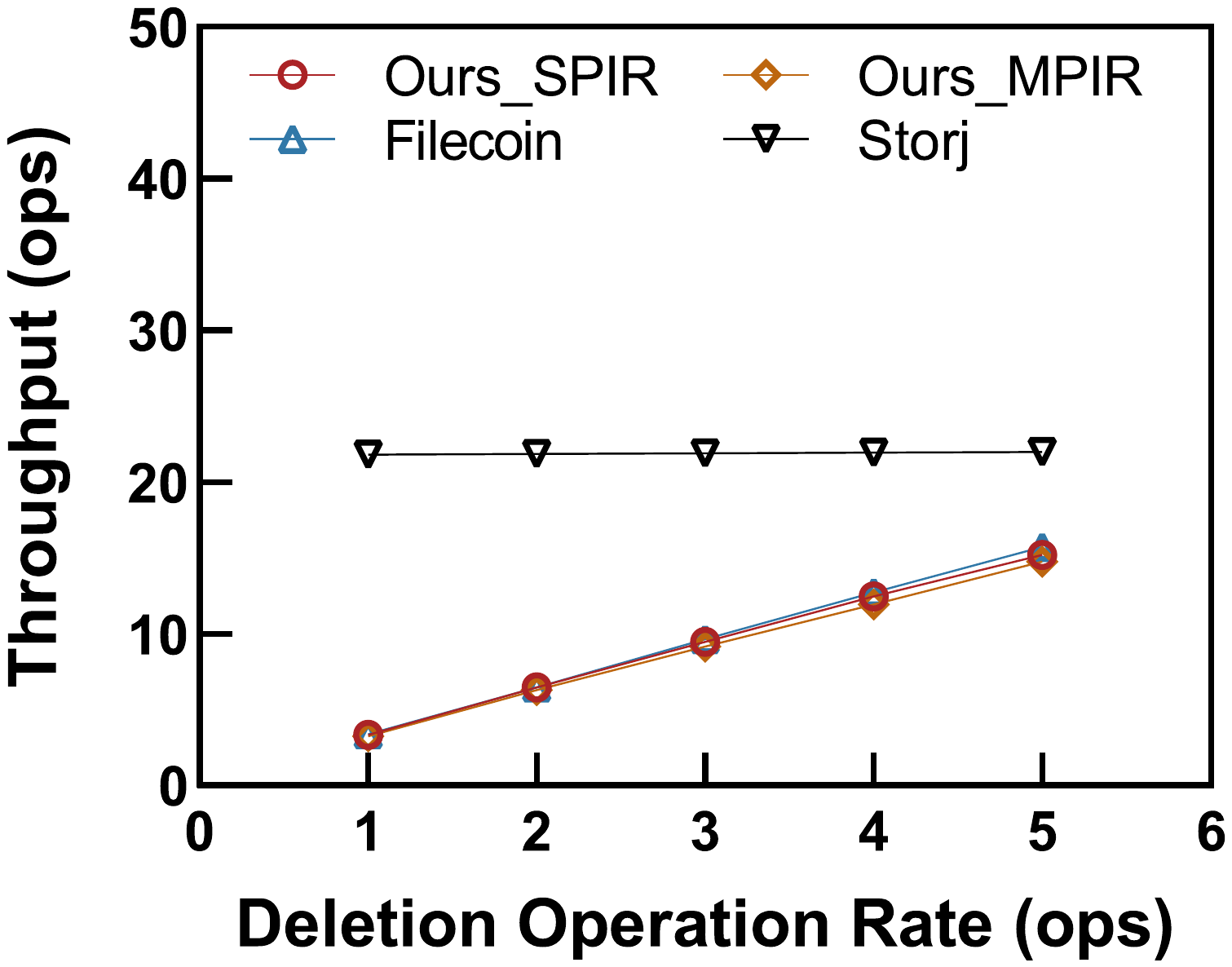}}
\subfigure[Latency]{
\label{Deletion Latency}
\includegraphics[width=0.46\linewidth]{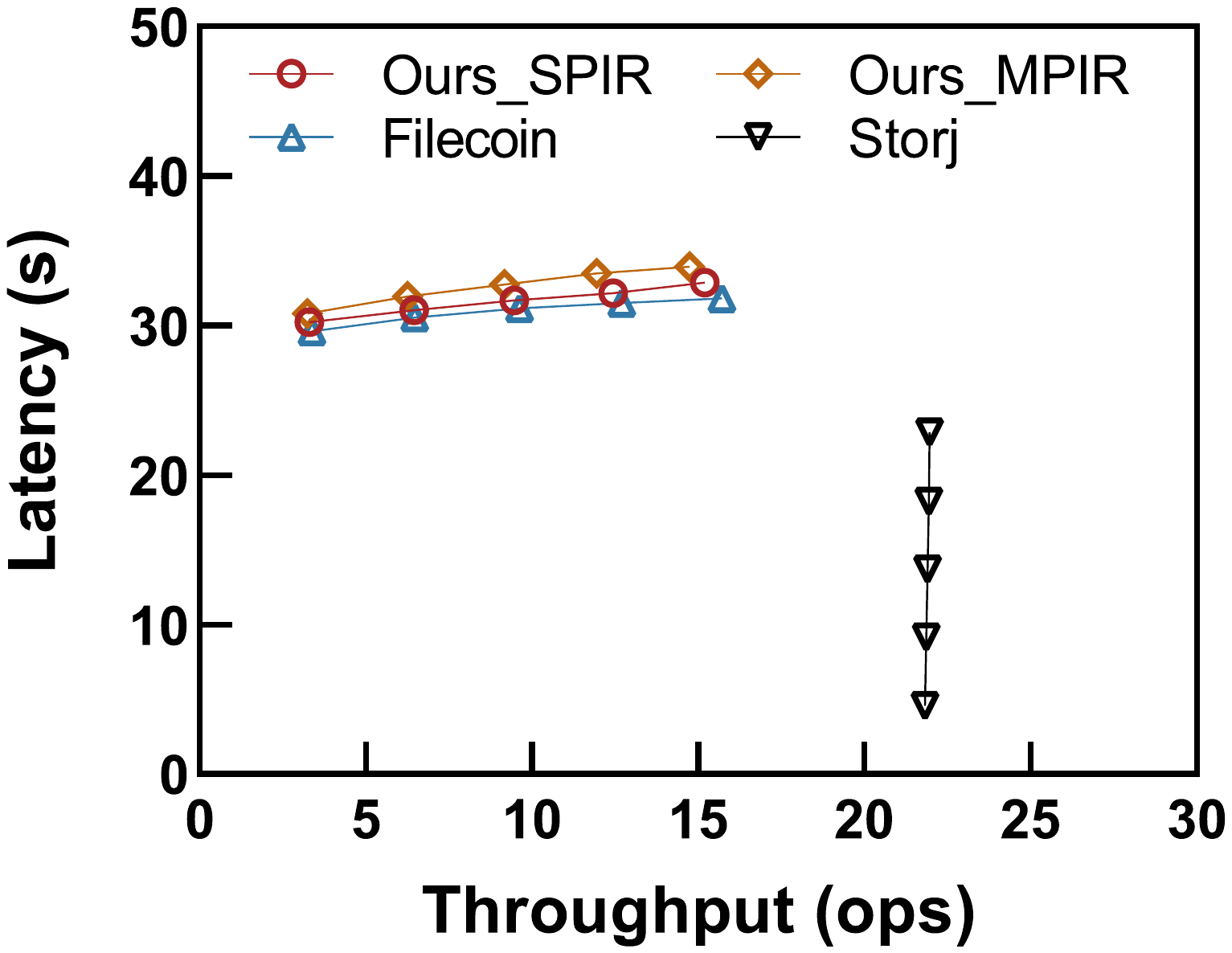}}
\caption{The throughput and latency of file deletion operation}
\label{Multi Deletion}
\end{figure}

\subsubsection{Retrieval Throughput and Latency} 
As depicted in Fig.~\ref{Retrieval Throughput} and Fig.~\ref{Retrieval Latency}, MPIR-DSN achieves slightly lower throughput and higher latency than Filecoin due to the extra computation of multi-server PIR, while SPIR-DSN shows even lower throughput and higher latency due to heavier overhead from single-server PIR.
Sia achieves a throughput of 32.78 Mbps, with latency increasing from 1.86 to 10.13 seconds. Storj reaches a throughput of 14.77 Mbps, with latency growing from 17.26 to 77.49 seconds.

\begin{figure}[!htbp]
\centering
\subfigure[Throughput]{
\label{Retrieval Throughput}
\includegraphics[width=0.46\linewidth]{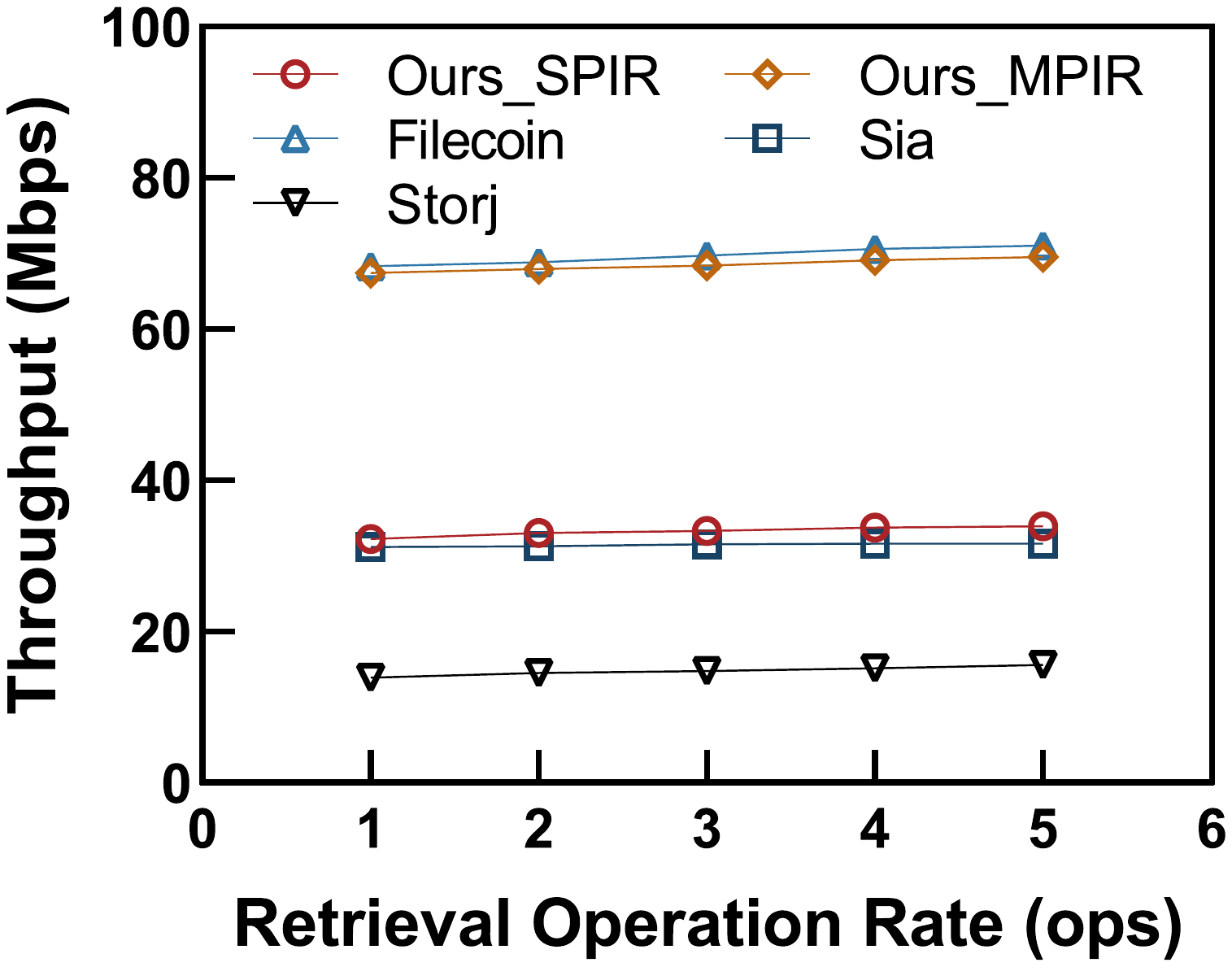}}
\subfigure[Latency]{
\label{Retrieval Latency}
\includegraphics[width=0.46\linewidth]{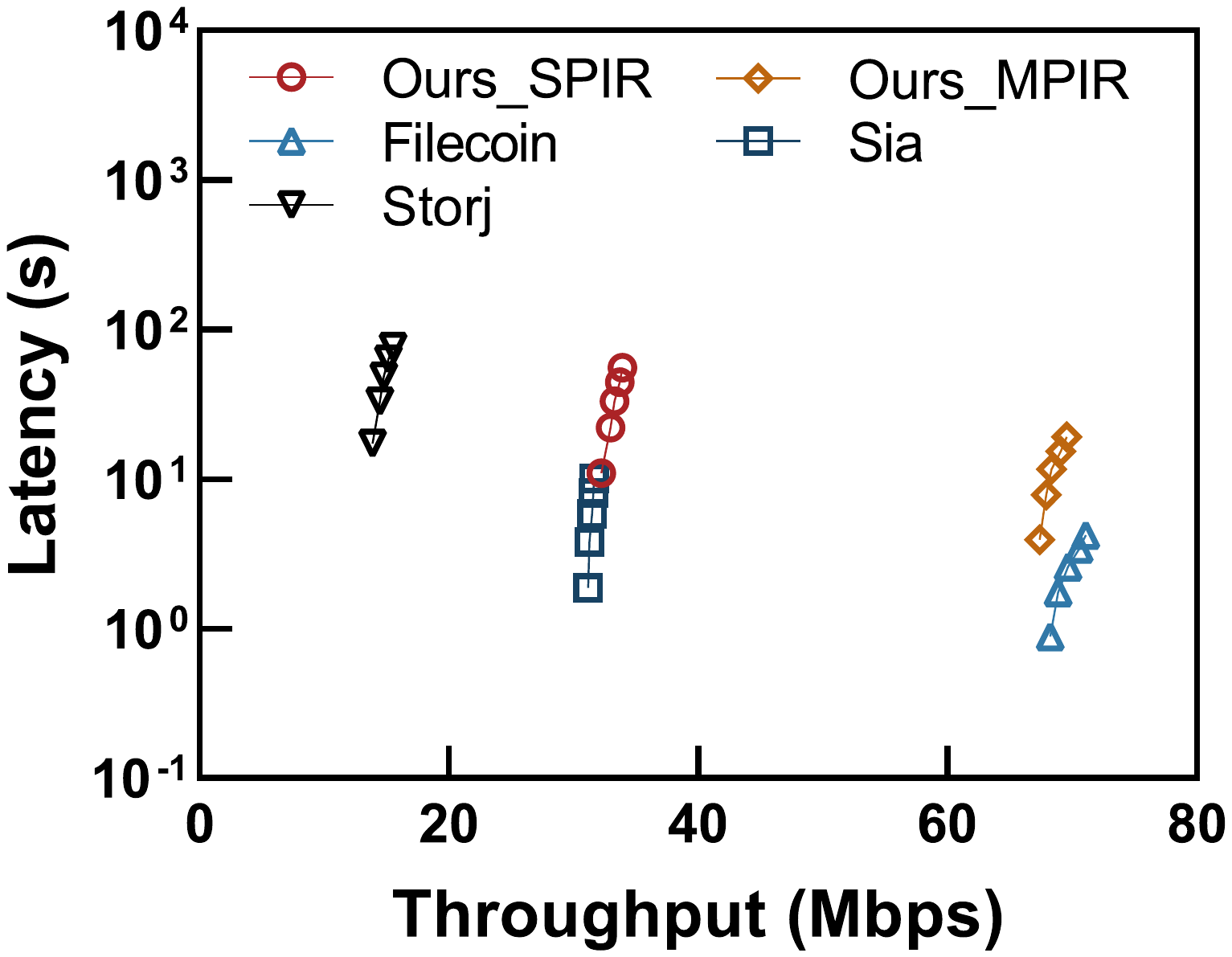}}
\caption{The throughput and latency of file retrieval operation}
\label{Multi Retrieval}
\end{figure}

\section{Conclusion}
We proposed PIR-DSN, the first DSN that supports private file retrieval. Leveraging an ACA, PIR-DSN compresses FIDs into varifiable contiguous indexes and enables private retrieval from both single and multiple miners. Our implementation shows PIR-DSN matches state-of-the-art DSNs in upload and deletion operations, while retrieval incurs extra cost from PIR computation. In future research, we aim to further improve the privacy and efficiency of operations in DSNs.



\section{Acknowledgment}
This study was supported by the National Key R\&D Program of China (No.2023YFB2703600), the National Natural Science Foundation of China (No. 62302266, 62232010, U23A20302, U24A20244), and the Shandong Science Fund for Excellent Young Scholars (No.2023HWYQ-008), 2025 Shandong Cyberspace Administration Reform Pilot Projects.


\bibliographystyle{IEEEtran}
\bibliography{reference}

\end{document}